\documentclass[11pt,letter]{article}
\usepackage{times,amsmath,amsfonts,stmaryrd,amsthm,amssymb,enumitem}
\usepackage{algorithm,booktabs,url,hyperref}
\usepackage[color=yellow]{todonotes}

\usepackage{algorithmic}

\usepackage{pifont}
\usepackage[round]{natbib}

\usepackage{tikz}
\usetikzlibrary{arrows}
\usetikzlibrary{calc}
\usetikzlibrary{shapes}
\usetikzlibrary{fit}
\usetikzlibrary{matrix} 
\usetikzlibrary{positioning}
\usetikzlibrary{decorations.pathreplacing}

\newtheorem{lemma}{Lemma}

\renewcommand{\eqref}[1]{Eq.~(\ref{eq:#1})}

\renewcommand{\P}{\mathbb{P}}
\newcommand{\E}{\mathbb{E}}

\newcommand{\reals}{\mathbb{R}}

\newcommand{\ceil}[1]{{\lceil #1\rceil}}

\newcommand{\one}{\mathbb{I}}

\newcommand{\cB}{\mathcal{B}}

\title{Selecting with History}
\author{Tom Hess and Sivan Sabato\\ Ben-Gurion University of the Negev\\ Beer Sheva, Israel\\\texttt{$\{$tomhe,sabatos$\}${@}cs.bgu.ac.il}}
\date{}

\newcommand{\totallength}{{N}}
\newcommand{\numsegments}{{K}}

\newcommand{\swh}{\ensuremath{\mathrm{SwH}}}

\newcommand{\Qinf}{{R}}

\newcommand{\psecp}{P_{\mathrm{sp}}}

\newcommand{\phasefrac}{\beta}
\newcommand{\phasesize}{\cB}

\setenumerate{leftmargin=*,align=left}

\usepackage[letterpaper,includeheadfoot,margin=1in]{geometry}

\begin{document}
\maketitle

\begin{abstract}
We define a new selection problem, \emph{Selecting with History}, which extends the secretary problem to a setting with historical information. We propose a strategy for this problem and calculate its success probability in the limit of a large sequence.
\end{abstract}

In the classical secretary problem \citep{dynkin1963optimum,gilbert1966recognizing}, $n$ numbers appear at a random order. The algorithm is allowed to select a single number. If it decides to select a number, it must do so immediately, before observing the next numbers, and it cannot later change its decision. The goal of the algorithm is to select the maximal number with the highest probability, where the set of numbers is selected by an adversary and the order of their appearance is random. \cite{gilbert1966recognizing} show that, for any input size $n$,
there is a number $t_n < n$ such that the optimal strategy is to observe the
first $t_n$ numbers, set $\theta$ to be the maximal number among those, and then
select the first number in the rest of the sequence  which is larger than
$\theta$. They show that $\lim_{n \rightarrow \infty} t_n/n = 1/e$ and that the probability of success of the optimal strategy by also tends to $1/e$ when for $n \rightarrow \infty$. 

In this note we define a new selection problem, \emph{Selecting with History} (SwH), which extends the secretary problem to a setting with historical information. We propose a strategy for this problem, and calculate its success probability in the limit of a large sequence.

Let $N,K \geq 2$ be integers, such that $K$ divides $N$. Let $Z$ be a finite
set of real numbers of size $N$. In this problem, the numbers in $Z$ are
ordered according to a uniformly random order. The algorithm observes
the first $N(1-1/K)$ numbers (the \emph{history}). Then, the algorithm
observes the last $N/K$ numbers (the \emph{selection sequence}) one by one,
and should select the maximal number in the selection sequence with the
highest probability. As in the secretary problem, the algorithm may only
select a number immediately after observing it, and cannot regret this
selection later.  The secretary problem is thus equivalent to $\swh$ with
$K = 1$. When $K \geq 2$, one can ignore the history and simply apply the
optimal secretary problem strategy to the selection sequence. However,
this does not exploit the information from the history. Instead, we propose the following strategy for
$\swh$. This strategy is parametrized by $\beta \in (0,1)$.
\begin{quote}
\emph{During the first $\ceil{\beta N/K}$ numbers in the selection sequence, select the first number that exceeds the $K$th-largest value in the history. If no such number was found in this part of the selection sequence, select from the rest of the sequence the first number that exceeds the maximal number observed so far in the selection sequence.}
\end{quote}
This strategy is inspired by a strategy proposed in \cite{gilbert1966recognizing} for a setting where a selection sequence is drawn i.i.d.~from a known distribution. Whereas under a known distribution the first threshold can be set based on this knowledge, here we estimate it based on the history.  

As in the secretary problem, the probability of success of this strategy depends only on the rank order of the numbers, and not on their specific values. For $K \geq 2$ that divides $N$, we denote by $\Qinf(N,K)$ the probability that the proposed strategy succeeds in selecting the maximal number from the selection sequence, for any $Z$ of size $N$. This probability depends on $\beta$, which we leave as an implicit parameter of $\Qinf$. 
For convenience, we also let $\Qinf(N,1) := \psecp(\totallength)$, where $\psecp(n)$ denotes the success probability of the optimal secretary problem strategy on an input sequence of size $n$.

Define 
\[
Q(K):=\lim_{L \rightarrow \infty} \Qinf(LK,K).
\]
By the definition of $\Qinf$, $Q(1) = \lim_{n\rightarrow \infty} \psecp(n) = 1/e$. 
The following lemma gives the value of $Q(\numsegments)$ for $K > 1$, as a function of $\beta$. 
\begin{lemma}\label{lem:qf}
If $K > 1$, then 
\begin{align*}
&Q(\numsegments) =\phasefrac \log(1/\phasefrac)(1-\frac{1}{\numsegments})^{\numsegments} + \sum_{j=1}^{\infty} \binom{j + \numsegments-1}{\numsegments-1} (1-\frac{1}{\numsegments})^{\numsegments} \frac{1}{\numsegments^{j}}\left(\frac{1-(1-\beta)^j}{j}+ \phasefrac \int_{\beta}^1 \frac{(1-x)^{j-1}}{x}\,dx\right). 
\end{align*}
\end{lemma}
\begin{proof}

Let $\totallength = L \numsegments$. We calculate $\Qinf(\totallength,\numsegments)$ based on its definition, and then take the limit $L \rightarrow \infty$. Let $Z = \{z_1,\ldots,z_N\}$ be the set of input numbers, where $z_1 > z_2 > \ldots > z_N$. Denote by $G$ the event that the \swh\ strategy selects the maximal number in the selection sequence, when the strategy is applied to a random ordering of $Z$. Let $A_1$ be the set of numbers in the history, and let $A_2 = Z \setminus A_1$ be the set of numbers in the selection sequence. Let $A_2' \subseteq A_2$ be the set of first $\phasesize := \ceil{\beta L}$ numbers in selection sequence. 
Let $\theta$ be the $\numsegments$'th largest number in $A_1$, and let $\theta'$ be the largest number in $A_2'$. The strategy described above selects the first number observed from $A_2'$ which is larger than $\theta$ if one exists. Otherwise, it selects the first number observed from $A_2 \setminus A_2'$ that is larger than $\theta'$ (if one exists). 
Let $J = |\{z \in A_2 \mid z > \theta\}|$. We have
$\P[G] = \sum_{j=0}^L \P[G \mid J=j]\P[J = j]$.
Note that the probabilities all depend (implicitly) on $L$.
Let $c > 1$. For any $L \geq cK^2$, 
\[
\P[G] = \sum_{j=0}^{cK^2} \P[G \mid J=j]\P[J = j] + \P[G \mid J > cK^2]\P[J > cK^2].
\]
Define $q(j) := \lim_{L\rightarrow \infty} \P[G \mid J=j]$, $p(j) := \lim_{L\rightarrow \infty} \P[J = j]$, and suppose that for some $\alpha:\reals\rightarrow \reals$, $\alpha(c) \geq \lim_{L\rightarrow \infty}\P[J > cK^2]$. Assuming all these limits exist, we have
\[
\sum_{j=0}^{cK^2}q(j)p(j) \leq \lim_{L\rightarrow \infty}  \P[G] \leq \sum_{j=0}^{cK^2}q(j)p(j) + \alpha(c).
\]
If in addition $\lim_{c \rightarrow \infty}\alpha(c)=0$, 
then, taking $c \rightarrow \infty$ on the inequality above, we get 
\begin{equation}\label{eq:pgfull}
\lim_{L\rightarrow \infty}  \P[G] = \sum_{j=0}^{\infty}q(j)p(j).
\end{equation}

We now give expressions for $q(j),p(j)$ and $\alpha(c)$.
First, for $p(j)$, we calculate $\P[ J=j]$.  Define the random variable $I$ which satisfies $\theta = z_I$. If $I = i$, this means that out of the numbers $z_1,\ldots,z_{i-1}$, exactly $\numsegments-1$ are in $A_1$, and also $z_{i} \in A_1$. Therefore $J = i - \numsegments$. Since $|A_1| = \totallength - L$ and its content is allocated uniformly at random, we have
\[
\P[J = i - \numsegments] = \P[I = i] = \binom{i-1}{\numsegments-1}\prod_{l=0}^{\numsegments-1} \frac{\totallength-L-l}{\totallength-l} \prod_{l=0}^{i-\numsegments-1}\frac{L-l}{\totallength-K-l}.
\]
Therefore 
\begin{equation*}
\P[J = j] = \binom{j + \numsegments-1}{\numsegments-1}\prod_{l=0}^{\numsegments-1} \frac{\totallength-L-l}{\totallength-l} \prod_{l=0}^{j-1}\frac{L-l}{\totallength-\numsegments-l}.
\end{equation*}
Taking the limit for $L \rightarrow \infty$ (recalling $\totallength = L \numsegments$) we get 
\begin{equation}\label{eq:pjlim}
p(j) \equiv \lim_{L\rightarrow \infty}\P[J = j]  = \binom{j + \numsegments-1}{\numsegments-1} (1-\frac{1}{\numsegments})^{\numsegments} \frac{1}{\numsegments^{j}}. 
\end{equation}

Second, to find $q(j)$, we now calculate $\P[G \mid J=j]$. If $j = 0$, then all $z \in A_2$ have $z < \theta$, therefore no element will be selected from $A_2'$. The probability of success is thus exactly as the probability of success of the secretary problem strategy with input size $L$ and threshold $\phasesize = \ceil{\beta L}$. Denote this probability $P_L$. We have, following the analysis in \cite{ferguson1989solved} for the secretary problem,
\[
\lim_{L\rightarrow \infty} P_L = \lim_{L \rightarrow \infty} \frac{\ceil{\phasefrac L}}{L} \sum_{i=\phasesize+1}^L \frac{1}{i-1} =  \lim_{L \rightarrow \infty} \frac{\ceil{\phasefrac L}}{L} \sum_{i=\phasesize+1}^L \frac{1}{L} \left(\frac{L}{i-1}\right) = \phasefrac \int_{\phasefrac}^1 \frac{1}{x}\, dx = \phasefrac \log(1/\phasefrac). 
\]
Hence 
,\begin{equation}\label{eq:pjzero}
q(0) \equiv \lim_{L\rightarrow \infty}\P[G \mid J=0] = \phasefrac \log(1/\phasefrac).
\end{equation}

To find $q(j)$ for $j > 0$, let $R$ be the location in $A_2$ of the maximal number $z_* = \max A_2$. Note that if the strategy does not select anything before reaching location $R$, it will certainly select $z_*$ by the definition of the strategy. Distinguish two cases:
\begin{enumerate}
\item If $R \leq \phasesize$, then $z_*$ is selected as long as all other $j-1$ items that exceed $\theta$ are located after $z_*$. Hence, for $r \leq \min(\phasesize, L-j+1)$
\[
\P[G \mid R = r, J= j] = \prod_{l=0}^{j-2} \frac{L-r-l}{L-1-l} = \prod_{l=0}^{j-2} (1 - \frac{r-1}{L-1-l}).
\]
\item If $R > \phasesize$, then $z_*$ is selected as long as all other $j-1$ items that exceed $\theta$ are located after $z_*$, and also the maximal item in the first $R-1$ items is in the first $\phasesize$ items, so that $z_*$ is the first item in $A_2 \setminus A_2'$ that is larger than $\theta'$. Hence, for $\phasesize \leq r \leq L-j+1$, 
\[
\P[G \mid R = r, J= j] = \prod_{l=0}^{j-2} \frac{L-r-l}{L-1-l}\frac{\phasesize}{r-1} = \prod_{l=0}^{j-2} (1 - \frac{r-1}{L-1-l})\frac{\phasesize}{r-1}.
\]
\item Neither of the conditions above can hold if $R > L -j +1$, since $j-1$ numbers cannot be located after $z_*$ in this case. Hence ,$\P[G \mid R > L -j +1, J= j] = 0$.
\end{enumerate}
Therefore
\begin{align*}
\P[G \mid J= j] &= \sum_{r=1}^{L-j-1}\P[R =r]\P[G \mid J= j, R = r]\\
&=\frac{1}{L} \sum_{r=1}^{L-j+1} \left(\frac{\phasesize}{r-1}\right)^{\one[r > \phasesize]}\cdot \prod_{l=0}^{j-2} (1 - \frac{r-1}{L-1-l}).
\end{align*}
We have
\[
(1 - \frac{r-1}{L+1-j})^{j-1} \leq \prod_{l=0}^{j-2} (1 - \frac{r-1}{L-1-l}) \leq (1 - \frac{r-1}{L-1})^{j-1},
\]

and $\frac{\phasesize}{r-1} = \frac{\ceil{\beta L}}{L}\frac{L}{r-1}$.
Therefore
\begin{align*}
\frac{1}{L} \sum_{r=1}^{L-j+1} \left(\frac{\ceil{\beta L}}{L}\frac{L}{r-1}\right)^{\one[r > \phasesize]}\cdot (1 - \frac{r-1}{L+1-j})^{j-1} &\leq \P[G \mid J= j]\\
&\leq \frac{1}{L} \sum_{r=1}^{L-j+1} \left(\frac{\ceil{\beta L}}{L}\frac{L}{r-1}\right)^{\one[r > \phasesize]}\cdot  (1 - \frac{r-1}{L-1})^{j-1}.
\end{align*}

Taking the limit $L \rightarrow \infty$ on both sides and defining $x = r/L$, this gives, for $j \geq 1$,
\begin{align}
q(j) \equiv \lim_{L \rightarrow \infty} \P[G \mid J= j] &= \int_{0}^\beta (1- x)^{j-1}\, dx + \phasefrac \int_{\beta}^1 \frac{(1-x)^{j-1}}{x}\,dx\notag\\
& = \frac{1-(1-\beta)^j}{j}+ \phasefrac \int_{\beta}^1 \frac{(1-x)^{j-1}}{x}\,dx.
\label{eq:qjnotzero}
\end{align}
Lastly, we are left to show an upper bound $\alpha(c) \geq  \lim_{L\rightarrow \infty}\P[J > cK^2]$ such that $\lim_{c \rightarrow \infty}\alpha(c)=0$.
Recall that if $\theta = z_i$ then $J = i -K$. 
For an integer $t$, denote $B_t = |\{i \mid i < t \text{ and } z_i \in A_2\}|$.
Note that $J \geq t$ if and only if $\theta \leq z_{t+K}$, which occurs if and only if $B_{t+K} \geq t$. 
Therefore $\P[J \geq t] = \P[B_{t+K} \geq t].$
We now give an upper bound on $\P[B_{t+K} \geq t]$, using a concentration bound for sampling without replacement from a population. Denote the ordered numbers in the selection sequence by $(z_{j_1},\ldots,z_{j_L})$. Then 
$B_t = \sum_{i=1}^{L} \one[j_i < t].$ 
$B_t$ is a sum of $L$ uniformly random draws without replacement from the sequence $x_1,\ldots,x_N$, where $x_i := \one[i<t]$. We have $\P[j_i < t] = (t-1)/N$. Hence ,$\E[B_t] = L(t-1)/N = (t-1)/K$. Setting $t = cK^2$ for $c > 1$, we have $t - (t+K-1)/K \geq t/4$. Hence 
,\[
\P[B_{t+K} \geq t] = \P[B_{t+K} - \E[B_{t+K}] \geq t - (t+K-1)/K] \leq 
\P[B_{t+K} - \E[B_{t+K}] \geq t/4].
\]
By Bernstein's inequality for sampling without replacement \citep{boucheron2013concentration}, setting $\epsilon =  t/(4L)$ and $\sigma^2 = \frac{1}{N}\sum_{i=1}^N (x_i-(t-1)/N)^2$,
\[
\P[B_{t+K} \geq t] \leq \P[\frac{1}{L}B_{t+K} - \frac{1}{L}\E[B_{t+K}] \geq \epsilon] \leq \exp(-L\epsilon^2/(2\sigma^2 + (2/3)\epsilon)).
\]
Noting that $\sigma^2 \leq \frac{1}{N}\sum_{i=1}^N x_i^2 = (t-1)/N \leq cK/L$, 
and $\epsilon = cK^2/(4L)$, we get that for some constant $b$, $\P[J \geq cK^2] \leq \exp(-b cK^2).$ 
Setting the RHS to $\alpha(c)$, we get $\alpha(c) \geq \lim_{L\rightarrow \infty}\P[J > cK^2]$ and $\lim_{c \rightarrow \infty}\alpha(c)=0$, as required. 
Combining \eqref{pgfull}, \eqref{pjlim}, \eqref{pjzero}, \eqref{qjnotzero} and the limit above, we get the equality in the statement of the lemma. 
\end{proof}
The value of $Q(K)$ for a given $\beta$ can be calculated numerically. We propose to select $\beta := 0.63$. This gives, e.g. $Q(2) \approx 0.47$, $Q(3) \approx 0.51$, $Q(10) \approx 0.55$. Compare this to $Q(1) = 1/e \approx 0.37$.

\vspace{1em}
In a previous version of this manuscript we proposed to use the strategy for SwH above to improve the competitive ratio of the submodular secretary problem under resource constraints. Unfortunately our analysis turned out to have an error which we have not been able to solve as of yet.

\bibliographystyle{abbrvnat}
\bibliography{arvix}

\end{document}